\documentclass[11pt]{amsart}

\usepackage{amsmath}
\usepackage{amsthm}
\usepackage{amssymb}
\usepackage{enumerate}
\usepackage{fullpage}
\usepackage[makeroom]{cancel}
\newtheorem{theorem}{Theorem}

\newtheorem{lemma}[theorem]{Lemma}

\newtheorem{proposition}[theorem]{Proposition}

\newtheorem{claim}[theorem]{Claim}

\newtheorem{ques}[theorem]{Question}
\newtheorem{conjecture}[theorem]{Conjecture}
 
\newcommand{\prob}[2]{\mathop{\mathrm{Pr}}_{#1}[#2]}

\newcommand{\defeq}{:=}

\newcommand{\sumv}{\Gamma}

\newcommand{\edges}{\mathbb{H}_{n}}
\newcommand{\bump}[2]{\text{\rm bump}_{#2}({#1})}
\newcommand{\swap}[2]{\text{swap}_{#2}({#1})}
\newcommand{\precset}[2]{S_{#2}(#1)}
\newcommand{\ceil}[1]{\lceil #1 \rceil}

\newcommand{\pisig}{\Pi_A(\sigma)}

\newcommand{\costpi}{c(\Pi)}
\newcommand{\e}{\text{\bf e}}

\begin{document}

\title{A Communication Game related to the Sensitivity Conjecture \footnotemark} 
\footnotetext{ A preliminary version of this paper appeared in
the Proceedings of the 6th Innovations in Theoretical Computer Science conference, 2015.}

\author{Justin Gilmer}
\address{Department of Mathematics,
Rutgers University,
Piscataway, NJ, USA.}
\email{jmgilmer@math.rutgers.edu}
\thanks{Supported by NSF  grant CCF 083727}
\author{Michal Kouck\'y}
\address{Computer Science Institute,
Charles University,
Prague, Czech Republic.}
\email{koucky@iuuk.mff.cuni.cz}
\thanks{The research leading to these results has received funding from the European Research Council under the European Union's Seventh Framework Programme (FP/2007-2013) / ERC Grant Agreement n. 616787. Partially supported by the project 14-10003S of GA \v{C}R.}
\author{Michael Saks}
\address{Department of Mathematics,
Rutgers University,
Piscataway, NJ, USA.}
\email{saks@math.rutgers.edu}
\thanks{Supported by NSF grants CCF-083727, CCF-1218711 and by  Simons Foundation award 332622.}

\maketitle

\begin{abstract}
One of the major outstanding foundational problems about boolean functions is the {\em sensitivity conjecture}, which
(in one of its many forms) 
asserts that the degree of a boolean function (i.e. the minimum degree of a real polynomial that interpolates the
function) is bounded above by some fixed power of its sensitivity (which is the maximum vertex degree of the graph
defined on the inputs where two inputs are adjacent if they differ in exactly one coordinate and their function
values are different).  We propose an attack on the sensitivity conjecture in terms of a novel two-player communication
game.  A lower bound of the form $n^{\Omega(1)}$ on the cost of this game would imply the sensitivity conjecture.

To investigate the problem of bounding the cost of the game, three natural (stronger) variants of the question are considered.  
For two of these variants,  protocols are presented that show that the hoped for lower bound does not hold.
These protocols satisfy a certain monotonicity property, and (in contrast to the situation for the two variants)
we show  that the cost of any monotone protocol satisfies a strong lower bound. 

There is an easy upper bound of $\sqrt{n}$ on the cost of the game.  We also improve slightly on this upper bound.
\end{abstract}

%\terms{Theory}

%\keywords{Sensitivity conjecture; degree of Boolean functions; sensitivity; decision trees; communication complexity}

\section{Introduction}

\subsection{A Communication Game}
The focus of this paper is a somewhat unusual cooperative two player communication game.
The game is parameterized by a positive integer $n$ and is denoted $G_n$.
Alice receives a permutation $\sigma=(\sigma_1,\ldots,\sigma_n)$
of $[n]=\{1,\ldots,n\}$ and a bit $b \in \{0,1\}$
and sends Bob a message  (which is restricted in a way that will be described momentarily).  Bob receives the
message from Alice and outputs a subset $J$ of $[n]$ that must include $\sigma_n$, the last
element of the permutation.   The cost to Alice and Bob is the size of the set $|J|$.

The message sent by Alice  is constrained as follows: Alice constructs
an array $\textbf{v}$ consisting of $n$ cells which we will refer to as {\em locations}, where each location $v_\ell$ is initially
empty, denoted by $v_\ell=*$. Alice gets the input as a data stream
$\sigma_1,\ldots,\sigma_n,b$ and is required to fill the cells of $\textbf{v}$ in the order specified by $\sigma$.
After receiving $\sigma_i$ for $i<n$, Alice
fills location $\sigma_i$ with  $0$ or $1$; once written this can not be changed.
Upon receiving $\sigma_n$ and $b$, Alice writes $b$ in location $\sigma_n$.
The message Alice sends to Bob is the completed array  in $\{0,1\}^n$.      

A protocol $\Pi$ is specified by Alice's algorithm for filling in the array, and Bob's function
mapping the received array to the set $J$.
The cost of a protocol $\costpi$ is the maximum of the output size $|J|$ over all
inputs $\sigma_1,\ldots,\sigma_n,b$.    

For example, consider the following protocol.  Let $k=\lceil \sqrt{n} \rceil$.  Alice and Bob fix a partition of
the locations of $\textbf{v}$ into $k$ blocks each of size at most $k$.
Alice fills $\textbf{v}$ as follows: When $\sigma_i$ arrives, if $\sigma_i$ is the last
location of its block to arrive then fill the entry with 1 otherwise fill it with 0.

Notice that if $b=1$ then the final array $\textbf{v}$ will have a single 1 in each block.  If $b=0$ then $\textbf{v}$
will have a unique all 0 block.

Bob chooses $J$ as follows: if there is an all 0 block, then $J$ is set to be that block, and otherwise
$J$ is set to be the set of locations containing 1's.  It is clear that $\sigma_n \in J$ and so
this is a valid protocol.  In all cases the size of $J$ will be at most $k$ and so the cost
of the protocol is $\lceil \sqrt{n} \rceil$. We will refer to this protocol as the AND-OR protocol. In Section \ref{sec:dec tree} we remark on this protocol's connection to the boolean function \[\text{AND-OR}(x) = \bigwedge\limits_{i = 1}^{\sqrt{n}}\bigvee\limits_{j = 1}^{\sqrt{n}} x_{ij}.\]

Let us define $C(n)$ to be the minimum cost of any protocol for $G_n$.
We are interested in the growth rate of $C(n)$ as a function of $n$.  In particular,
we propose:

\begin{ques}
\label{question:main}
Is there a $\delta>0$ such that $C(n) = \Omega(n^{\delta})$?
\end{ques}

\subsection{Connection to the Sensitivity Conjecture}

Why consider such a strange game?    The motivation is that the game provides a possible approach to the well known {\em
sensitivity conjecture} from boolean function complexity.  

Recall that the sensitivity of an $n$-variate boolean function $f$ at an input $\textbf{x}$, denoted $s_{\textbf{x}}(f)$, is the number of locations $\ell$
such that if we flip the bit of $\textbf{x}$ in location $\ell$ then the value of the function changes.  (Alternatively, this is the number
of neighbors of $\textbf{x}$ in the hamming graph whose $f$ value is different from $f(\textbf{x})$.)  The sensitivity of $f$,
$s(f)$,  is the maximum of $s_{\textbf{x}}(f)$ over all boolean inputs $\textbf{x}$.

The degree of a function $f$, $deg(f)$, is the smallest degree of a (real) polynomial $p$ in variables $x_1,\ldots,x_n$ that
agrees with $f$ on the boolean cube.  %It is easy to show and well known that $deg(f) \geq s(f)$.  

\begin{conjecture} (The Sensitivity Conjecture) There is a $\delta>0$ such that for any boolean function $f$,
 $s(f) \geq \Omega(deg(f)^{\delta})$.
\end{conjecture}

An easy argument (given in Section \ref{sec:connection}) connects the cost function $C(n)$ of the game $G_n$ to the sensitivity conjecture:

\begin{proposition} \label{prop:connection}
For any boolean function on $n$ variables, $s(f) \geq C(deg(f))$. 
\end{proposition}

In particular, an affirmative answer to Question \ref{question:main} would imply the sensitivity conjecture.

We note that Andy Drucker \cite{drucker} independently formulated the above communication game, and observed its connection
to the sensitivity conjecture.

\subsection{Background on the Sensitivity Conjecture}

Sensitivity and degree belong to
a large class of complexity measures for boolean functions
that seek to quantify, for each
function $f$, 
the amount of knowledge about individual variables needed to evaluate  $f$. 
Other such measures include decision tree complexity and
its randomized and quantum variants, certificate complexity,
and block sensitivity.    The value of such a  measure is at most the number
of variables.   
There is a long line of research aimed at bounding one such measure in terms of another.  For 
measures $a$ and $b$ let us write $a \leq_r b$ if there are constants $C_1$, $C_2$
such that for every total boolean function $f$, $a(f) \leq C_1 b(f)^r+C_2$. 
For example, the decision tree complexity of $f$, $D(f)$, is at least its degree $deg(f)$ and thus
$deg \leq_1 D$. It is also known \cite{midrijanis} that  $D \leq_3 deg$.
We say that
$a$ is {\em polynomially bounded} by $b$ if $a \leq_r b$ for some $r>0$ and that $a$ and $b$ are
{\em polynomially equivalent} if each  is polynomially bounded by the other. 

The measures mentioned above, with the notable
exception of sensitivity, are known to be polynomially equivalent. For example, Nisan and Szegedy \cite{nisanszegedy} proved $bs(f) \leq_2 deg(f)$, and also proved a result in the other direction, which was improved in ~\cite{beals2001quantum} to $deg(f) \leq_3 bs(f)$. 
 For a survey of such results, see \cite{buhrmandewolf} and 
\cite{senssurvey}; some recent results include \cite{abblss15,gss13}.

 The sensitivity conjecture, posed as a question by Nisan \cite{nisan89}
asserts that $s(f)$ is polynomially equivalent to at least one (and therefore all)  of the other measures mentioned.  There are many reformulations and related conjectures;
see \cite{senssurvey} for a survey. 

\iffalse
There are a number of equivalent formulations of the sensitivity conjecture. For instance \cite{gotsman1992equivalence} give a graph theoretic formulation by exploring a different relationship between sensitivity and degree than what is presented here. The same graph theoretic question also appeared somewhat earlier in \cite{chung1988induced}, however, sensitivity of boolean functions was only mentioned as a related problem and no direct connection was given. For a good survey of many other variations of the sensitivity conjecture, see \cite{senssurvey}. 
\fi

 The sensitivity conjecture perhaps more commonly appears as a question on the relationship between sensitivity and block sensitivity. For example, Nisan and Szegedy \cite{nisanszegedy} asked specifically if $bs(f) = O(s^2(f))$ for all functions, and as of this writing no counterexample has been given. The best known bound relating sensitivity to another measure was given by Kenyon and Kutin \cite{kenyon2004sensitivity}. They proved that $bs(f) \leq \frac{e}{2\pi}e^{s(f)}\sqrt{s(f)}$ for all boolean functions.%, where $bs(f)$ and $s(f)$ denote the block sensitivity and sensitivity of $f$ respectively. %Although this work proves no new bounds between sensitivity and a known measure, we remark that proving Theorem \ref{thm:oblivious} in a slightly more general setting would imply new bounds between sensitvity and degree.

%Like this work, \cite{gotsman1992equivalence} also use a relationship between sensitivity and degree to explore another approach to the conjecture. However, they prove a graph theoretic formulation of the conjecture, and the approach seems otherwise unrelated to this work. %Although this work shows no new bounds relating sensitivity to another known measure, we feel that the partial results in this work demonstrate the promise of the new approach.
\subsection{Outline of the Paper}
  In Section \ref{sec:connection} we prove that a positive answer to Question \ref{question:main} would imply the sensitivity conjecture. We show that adversary arguments for proving that boolean functions are evasive (that is have decision tree complexity $D(f) = n$) provide strategies for the communication game.  We also prove that it suffices to answer Question \ref{question:main} for the restricted class of  {\em order oblivious protocols}.
  
  In Section \ref{sec:prob} we present three stronger variants of Question \ref{question:main}. We exhibit protocols
that show that two of these variants have negative answers.  One might then expect that
variants of one of these protocols might lead to a negative answer to Question \ref{question:main}. 
 However,   we observe that these protocols
satisfy a property called monotonicity and 
 in Section \ref{sec:lowerbounds} we prove a $\lfloor \sqrt{n}\rfloor$ lower bound on the cost of any monotone
protocol. Thus a protocol that gives a negative answer to Question~\ref{question:main} must look quite different
from the two protocols that refuted the strengthenings.  We also prove a rather weak lower bound for a special class of protocols called
assignment oblivious protocols.   Finally, in Section \ref{sec:bestknown} we construct a protocol with  cost  $.8\sqrt{n}$, thus beating
the AND-OR protocol by a constant factor.   Let $r(k)=\log(C(k)/\log(k))$. After a preliminary version of our
paper appeared,
Szegedy \cite{szegedy15} showed that for any $k$,  $C(n) =O(n^{r(k)})$.
Our example shows that there is a $k$ for which $r(k)<1/2$ and so  it follows that 
$C(n) =O(n^{1/2-\delta})$ for some $\delta>0$.  Szegedy further showed
that $C(30)\leq 5$ which gives the best currently known upper bound $C(n) = O(n^{0.4732})$.

\section{Connection between the Sensitivity Conjecture and the Game}\label{sec:connection}
   In this section we prove Proposition \ref{prop:connection}, which connects the sensitivity conjecture with the two player game described in the introduction. 
   
   We use $\e_\ell$ to denote the assignment in $\{0,1\}^n$ that is 1 in location $\ell$ and 0 elsewhere. Given  $\textbf{v},\textbf{w} \in \{0,1\}^n$,  $\textbf{v} \oplus \textbf{w}$  denotes their bitwise mod-2 sum. 

Alice's strategy maps the permutation-bit pair $(\sigma,b)$
to a boolean array $\textbf{v}$ and Bob's strategy maps the array $\textbf{v}$ to a subset of $[n]$.  We now show
that for each strategy for Alice there is a canonical best strategy for Bob.    
For a permutation $\sigma$, $\Pi_A(\sigma)$  denotes 
the array Alice writes while receiving  $\sigma_1, \cdots, \sigma_{n-1}$ 
(so location $\sigma_n$ is  labeled  $*$). Thus $\pisig$ can be viewed as an 
edge in the {\em hamming graph} $\edges$ whose vertex set
is $\{0,1\}^n$, with two vertices adjacent if they differ in one coordinate.   
The {\em edge set} $E(\Pi)$ of a protocol $\Pi$ is the set of  edges $\pisig$ over all permutations $\sigma$.
This defines a subgraph of $\edges$.  Given Alice's output $\textbf{v}$, the possible values for $\sigma_n$
are precisely those locations $\ell$ that satisfy $(\textbf{v}, \textbf{v} \oplus \e_\ell)$ is an edge in $E(\Pi)$.  Thus the best
strategy for Bob is to output this set of locations.  It follows that 
$\costpi$ is equal to the maximum vertex degree of the graph $E(\Pi)$. 
   
Proposition~\ref{prop:connection} will therefore follow by showing the following: Given a boolean function 
with degree $n$ and sensitivity $s$, there is a strategy $\Pi$ for Alice for the game $G_n$ such that
the graph $E(\Pi)$ has maximum degree at most $s$.  

We need a few preliminaries.
A {\em subfunction} of a boolean function $f$ is a function $g$ obtained from $f$ by fixing some of the variables of $f$ to 0 or 1. For a subfunction $g$ of $f$, $s(f) \geq s(g)$. We say a function has {\em full degree} if $deg(f)$ is equal to the number of variables of $f$. We start by recalling some well known facts.
   \begin{lemma}
      For any boolean function $f$ there exists a subfunction $g$ on $deg(f)$ variables that has full degree. 
   \end{lemma}
   \begin{proof}
   If $p$ is the (unique) multilinear real polynomial that agrees with $f$ on the boolean cube, then $p$ contains a monomial $\prod\limits_{\ell \in S} x_\ell$ where $|S| = deg(f)$. Let $g$ be the function obtained by fixing the variables in $[n]\setminus S$ to 0. Then $g$ is a function on $deg(f)$ variables that has full degree.
   \end{proof}

   \begin{lemma} \label{lem:restriction}
   Given a function $f$ with full degree and a location $\ell$, there exists a bit $b$ such that the function obtained from $f$ by fixing $x_\ell = b$ is also of full degree. 
   \end{lemma}
   \begin{proof}
   The polynomial (viewed as a function from $\{0,1\}^n \rightarrow \{0,1\}$) for $f$ may be written in the form $p_1(x_1,x_2,\cdots,\cancel{x_\ell},\cdots, x_n) + x_\ell p_2(x_1,x_2,\cdots,\cancel{x_\ell},\cdots,x_n)$. Here $p_1(x_1,x_2,\cdots,\cancel{x_\ell},\cdots, x_n)$ indicates that the variable $x_\ell$ is not an input to the polynomial. If $p_1$ has a non zero coefficient on the monomial $\prod\limits_{k \neq \ell}x_k$, then we set $x_\ell = 0$ and the resulting function will have full degree. For the other case, note $p_2$ must have a non zero coefficient on $\prod\limits_{k \neq \ell}x_k$ because $f$ has full degree. Thus, setting $x_\ell = 1$ will work. 
   \end{proof}

 The proof of this lemma is essentially the same as the standard argument
that the decision tree complexity of any function $f$ is at least  $deg(f)$. \ 
   
   We are now ready to prove Proposition \ref{prop:connection}.
   \begin{proof}
   Given $f$, let $g$ be a subfunction on $deg(f)$ variables with full degree. We  construct a protocol $\Pi$ that satisfies $E(\Pi) \subseteq E(g)$, where $E(g)$ denotes the set of sensitive edges for the function $g$, i.e. the edges of $\edges$
whose endpoints are mapped to different values by $g$, which implies $\costpi \leq s(g) \leq s(f)$, and thus proves the proposition. As Alice receives  $\sigma_1,\sigma_2,\cdots,\sigma_n$, she fills in $\textbf{v}$  so that the restriction of $f$ to
each partial sucessive partial assignment  remains
a full degree function, which is possible by Lemma \ref{lem:restriction}. 
After Alice fills location $\sigma_{n-1}$, the function $g$ restricted to $\textbf{v}$ is a non-constant 
function of one variable, and so the edge $\pisig$ is a sensitive edge for $g$. This implies that $E(\Pi) \subseteq E(g)$.  %Thus if $H$ is the subgraph of the boolean cube which is the union of all edges $A(\sigma)$ for $\sigma \in S_{deg(f)}$, then every edge in $H$ is a sensitive edge for $g$ and $s(g) \geq \Delta(H)$. 
   
   %The adversary $A$ can be used to construct a protocol $\Pi$ for the two player game $G_{deg(f)}$. As Alice is being streamed a permutation $\sigma$ she answers according to the adversary $A$. When Bob receives the array $v$ he responds with $J$ which is the set if indices $i$ for which $(v, v \oplus i)$ is a sensitive edge for $g$. Since $v$ was arrived by following the adversary $A$, it follows that $(v, v \oplus \sigma_n)$ is a sensitive edge for $g$ and thus $\sigma_n \in J$. Thus this is a valid protocol. Note the cost of this protocol is at most $s(g)$, because $|J|$ is always the number of sensitive edges at a vertex $v$. This implies  $s(g) \geq C(deg(f))$.
   \end{proof}
The proof shows that a degree $n$ Boolean function
having sensitivity $s$ can be converted into a strategy
for Alice for the game $G_n$ of cost at most $s$.  We don't know whether
this connection goes the other way, i.e., we can't
rule out the possibility that the answer to Question~\ref{question:main} is negative (there
is a very low cost protocol for $G_n$) but
the sensitivity conjecture is still true. 

   \subsection{Connection to Decision Tree Complexity} \label{sec:dec tree}
   %We note the connection between protocols $\Pi$ for the game $G_n$ and boolean functions on $n$ variables for which $D(f) = n$ (sometimes referred to as {\em evasive} functions). A common method in the literature for showing that a function is evasive is to construct an adversarial strategy for responding to a data stream $\sigma_1,\sigma_2,\cdots,\sigma_n$, for which the sequence of partial assignments always avoids any certificates for the function $f$. The protocol given in the introduction is essentially the adversary strategy which may be used to show that the TRIBES function is evasive (if you switch the 0's and 1's in the protocol then it becomes exactly the adversary). Furthermore, it holds that a function $f$ is evasive if and only if there exists a protocol $\Pi$ for which $E(\Pi) \subseteq E(f)$. 
An $n$-variate boolean function is {\em evasive} if its decision tree complexity is $n$.   
A common method for proving evasiveness  is via an {\em adversary argument}.  View the problem of evaluating the function by a decision tree
as a game between the querier who wishes to evaluate the function and who decides which variable to read next, and the adversary who decides the value of the variable.   A function is evasive if
there is a strategy for the adversary that forces the querier to ask all $n$ quesitons.  For example, to prove that
   \[\text{AND-OR}(\textbf{x}) = \bigwedge\limits_{i = 1}^{\sqrt{n}} \bigvee\limits_{j = 1}^{\sqrt{n}} x_{ij}\]
is evasive,  the adversary can use the strategy: answer 0 to every variable unless the variable is the last variable in its $\bigvee$-block, in which case answer 1.
This adversary is exactly Alice's part of the AND-OR protocol described in the introduction. For more examples of adversary arguments see \cite{lovasz2002lecture}.
     
     Every evasive function $f$ by definition admits an adversary argument, and this corresponds to a protocol $\Pi$ for Alice.   
In fact a function $f$ is evasive if and only if there exists a protocol $\Pi$ for which $E(\Pi) \subseteq E(f)$ (recall $E(f)$ is the set of sensitive edges of the function $f$), and thus
the cost (size of the set chosen by Bob) is at most the sensitivity of $f$. 
This work explores whether we can use the structure of an arbitrary adversary (or protocol) to exhibit a lower bound on sensitivity.
   \subsection{Order Oblivious Protocols}

In the game $G_n$, at each step $i<n$,  the value written by Alice at location $\sigma_i$ may depend on her knowledge up
to that step, which includes both the sequence
$\sigma_1,\cdots,\sigma_{i}$ and the partial assignment already made to $v$ at locations
$\sigma_1,\ldots,\sigma_{i-1}$. 
A natural way to restrict Alice's strategy is to require that the bit she writes in location $\sigma_i$ depends only
on $\sigma_i$ and the current partial assignment to $v$ but not on the order in which $\sigma_1,\ldots,\sigma_{i-1}$
arrived. A protocol satisfying this restriction is said to be {\em order oblivious}. The following easy proposition shows that it suffices to answer Question \ref{question:main} for order oblivious protocols.

\begin{proposition} \label{prop:orderoblivious}
   Given any protocol $\Pi$ there exists an order oblivious protocol $\Pi'$ such that $E(\Pi') \subseteq E(\Pi)$. In particular, $c(\Pi') \leq c(\Pi)$.
\end{proposition}
\begin{proof}
  First some notation. Given a permutation $\sigma$ let $\sigma_{\leq k}$ denote the prefix of the first $k$ elements of $\sigma$. We let $\Pi_A(\sigma_{\leq k})$ denote the partial assignment written on $\textbf{v}$ after Alice has been streamed $\sigma_1,\cdots,\sigma_k$.
  
Given $\Pi$ we define an order oblivious protocol $\Pi'$ of cost at most that of $\Pi$.  We define $\Pi'$ in steps, (where in step $i$ Alice receives $\sigma_i$ and writes a bit in that location).  
Given $k \geq 0$ we assume that $\Pi'$ has been defined up through step $k$ and has the property
that for every permutation $\sigma$, there is a permutation $\tau$ of $\sigma_1,\cdots,\sigma_k$ so that $\Pi_A(\tau) = \Pi_A'(\sigma_{\leq k})$.
  
  Suppose $\sigma_{k+1}$ arrives and the current state of the array is $\textbf{v} \defeq \Pi'(\sigma_{\leq k})$.  From $\textbf{v}$ Alice can deduce the set $\{\sigma_1,\ldots,\sigma_k\}$  (the set of locations not labeled  *). Alice then considers all permutations $\tau$ of $\sigma_1,\cdots,\sigma_k$ such that $\Pi_A(\tau) = \Pi_A'(\sigma_{\leq k})$ and picks the lexicographically smallest permutation (call it $\tau^*$) in that set and writes on location $\sigma_{k+1}$ according to what $\Pi$ does after $\tau^*$. Note that the bit written on location $\sigma_{k+1}$ does not depend on the relative order of $\sigma_1,\sigma_2,\cdots,\sigma_k$. 
%Using this strategy, Alice maintains the invariant that for every permutation $\sigma$, there is a permutation $\tau$ of $\sigma_1,\cdots,\sigma_k$ so that $\Pi(\tau) = \Pi'(\sigma_{\leq k})$.
   
By construction, $\Pi'$ is order oblivious. Also for any permutation $\sigma$ there is a permutation $\tau$ for which $\Pi_A(\tau) = \Pi_A'(\sigma)$. This implies that $E(\Pi') \subseteq E(\Pi)$.
   \end{proof}
   
\section{Stronger Variants of\\
  Question~1}%\ref{question:main}}
\label{sec:prob}

We now present three natural variants of Question~\ref{question:main}, and refute two of them by exhibiting and analyzing some specific protocols.

The cost function $\costpi$ of a protocol is the worst case cost over all choices of $\sigma_1,\ldots,\sigma_n,b$.  Alternatively, we can consider
the average size (with respet to random $\sigma$ and $b$) of the set Bob outputs.  We call this the {\em expected cost of $\Pi$} and denote it by $\tilde{c}(\Pi)$.  Let $\tilde{C}(n)$ denote the minimum
expected cost of a protocol for $G_n$.  

\begin{ques}
\label{question:average}
Is there a $\delta>0$ such that $\tilde{C}(n) = \Omega(n^{\delta})$?
\end{ques}

An affirmative answer to this question would give an affirmative answer to Question \ref{question:main}.

It is well known that the natural probabilistic version of the sensitivity conjecture, where sensitivity is replaced
by average sensitivity (wiith respect to the uniform distribution over $\{0,1\}^n$) is trivially false (for example, for the
OR function).  However, there is apparently no connection between average sensitivity and average protocol cost.   For example,  the protocol induced by the decision tree adversary for OR has Alice write a 0 at each step.   Note  that $E(\Pi)$ is  exactly the set of sensitive edges for the OR function. However, the average cost $\tilde{c}(\Pi)$ is $n/2$ whereas the average sensitivity of the OR function is $o(1)$.     

We also remark that an analog of Proposition~\ref{prop:orderoblivious} holds for the cost function $\tilde{c}(\Pi)$, and therefore
it suffices to answer the question for order oblivious protocols.  (The proof of the analog is similar to the proof of Proposition~\ref{prop:orderoblivious}, except when modifying the protocol  $\tau^*$ is not selected to be the lexicographically smallest permutation
in the indicated set, but rather the permutation in the indicated set that minimizes the expected cost conditioned on
the first $k$ steps. )

There is another natural variant of Question~\ref{question:main} based on average case.  When we run a fixed protocol
$\Pi$ on  a random permutation $\sigma$ and bit $b$, we can view the array $\textbf{v}$ produced
by Alice as a random variable. Let $\tilde{h}(\Pi)$ be the conditional entropy
of $\sigma_n$ given $\textbf{v}$; intuitively this measures the average number of bits of uncertainty that Bob
has about $\sigma_n$  after seeing $\textbf{v}$.  It is easy to show that this is bounded above
by $\log(\costpi)$.  
Let $\tilde{H}(n)$ be the minimum
of $\tilde{h}(\Pi)$ over all protocols $\Pi$ for $G_n$.
The analog  of Question \ref{question:main} is whether there is a positive constant $\delta$ such that
$\tilde{H}(n) = \Omega(\delta \log(n))$.   An affirmative answer to this would have implied an affirmative
answer to Question~\ref{question:main}, but  the answer to this new question turns out to be  negative.

\begin{theorem} \label{thm:construction} 
There is an order oblivious protocol $\Pi$ for $G_n$ such that $\tilde{h}(\Pi)=3+\lceil \log\log(n) \rceil$.
\end{theorem}  
{\bf Remark:} It might seem that this could be proved  by giving a protocol $\Pi$ that is not order
oblivious and converting it  into an order oblivious protocol  as described earlier..
However, while we know that this can be done without increasing worst case cost or  average cost,
it is possible that $\tilde{h}$ may inrease. Therefore, we  construct the desired order oblivious protocol directly.
\begin{proof}
 Let $k=\lceil \log(n) \rceil$ and associate each location $\ell \in [n]$ to
its binary expansion, viewed as a vector $\textbf{b}(\ell)  \in \mathbb{F}_2^k$. Note that $0 \notin [n]$, and thus each vector $\textbf{b}(\ell)$ is nonzero. 
For an array $\textbf{v} \in \{0,1\}^n$ we define $\sumv(\textbf{v})$ to be $\sum_{i=1}^n \textbf{b}(\ell)$, i.e. the vector in $\mathbb{F}_2^k$
obtained by summing the vectors
corresponding to the 1 entries of $\textbf{v}$.  Say that an array $\textbf{v} \in \{0,1,*\}^n$ is {\em admissible} if there is a way of
filling in the *'s (a {\em completion}) so that for the resulting array $\textbf{w}$ we have $\sumv(\textbf{w})=0^k$, where $0^k$ is the all 0 vector in $\mathbb{F}_2^k$.
For an admissible array $\textbf{v}$, let $\hat{\textbf{v}}$ be the unique completion of $\textbf{v}$ such that (1) $\sumv(\hat{\textbf{v}}) = 0^k$,
(2) The number $r$  of 0's in $\hat{\textbf{v}}$ is minimum, (3) the ordered sequence $\ell_1<\cdots < \ell_r$
of locations of the 0's in $\hat{\textbf{v}}$ is lexicographically minimum subject to 
conditions (1) and (2), i.e., for each $j \in [r]$, $\ell_j$ is minimum possible given $\ell_1,\ldots,\ell_{j-1}$.

We now describe the protocol.
 Let $t>k$ be an integer (which
we'll choose to be $\lceil \log^2(n) \rceil$).  Alice says 0 for the first $n-t$ steps.  The resulting array $\textbf{u}$ has $n-t$ 0's
and $t$ $*$'s.  Since $\textbf{u}$ can be completed to the all 0 array,  $\textbf{u}$ is admissible.  Furthermore, among the $*$ positions
there must be a set of at most $k$ vectors that sum to $0^k$, so $\hat{\textbf{u}}$ has at most $k$ 1's. 
Alice fills in the remaining positions to agree with $\hat{\textbf{u}}$.
This strategy is order oblivious:   a simple induction shows that
for each array $\textbf{w}$ reached under the above strategy,
 $\textbf{w}$ is admissible and
$\hat{\textbf{w}}=\hat{\textbf{u}}$, so Alice's strategy is equivalent to
filling position $\sigma_k$ (for $k \geq n-t$)  according to $\hat{\textbf{w}}$ where $\textbf{w}$ is the array after $k-1$ steps. This
is clearly an order oblivious strategy

Let $\textbf{v}$ denote the array in $\{0,1\}^n$ received by Bob.
We now obtain an upper bound on the conditional entropy of $\sigma_n$ given $\textbf{v}$.  Let $\textbf{u}=\textbf{u}(\sigma)$
be the array obtained after the first $n-t$ steps and let $T(\sigma)$ be the set of positions of *'s in $\textbf{u}$.
Let $S(\sigma)$ be the subset of $T(\sigma)$ consisting of those positions set to 0 in $\hat{\textbf{u}}$.
Let $L$ be the random variable that is 1 if $\sigma_n \in S(\sigma)$ and 0 otherwise.  Since $S(\sigma)$ depends
only on the set $T(\sigma)$ and not on the order of the last $t$ locations, the probability
that $L=1$ is $|S(\sigma)|/|T(\sigma)| \leq \log(n)/\log^2(n) = 1/\log(n)$.
We have:

\begin{eqnarray*}
H(\sigma_n|\textbf{v}) & \leq & H(\sigma_n,L|\textbf{v}) \\
& = & H(L|\textbf{v})+H(\sigma_n|\textbf{v},L) \\
& \leq & 1 + H(\sigma_n|\textbf{v},L) \\
&=& 1+ H(\sigma_n|\textbf{v},L=1)\prob{}{L=1} \\
 & & +H(\sigma_n|\textbf{v},L=0)\prob{}{L=0}\\
& \leq & 1 + H(\sigma_n)\frac{1}{\log(n)} + H(\sigma_n|\textbf{v},L=0)
\end{eqnarray*}

We bound the final expression. $H(\sigma) = \log(n)$ so the second term is  1.
For the third term, we condition further on the value of the final bit $b$:

\begin{eqnarray*}
 H(\sigma_n|\textbf{v},L=0) & \leq & H(b) + \frac{1}{2} (H(\sigma_n|\textbf{v},L=0,b=1) + H(\sigma_n|\textbf{v},L=0,b=0))\\
\end{eqnarray*}
Of course, $H(b)=1$.  
Given $L=0$, we have $\sigma_n \in T(\sigma)-S(\sigma)$.   If $b=1$, then $\sigma_n$ is one of at most
$t$ positions set to 1, and so the conditional entropy of $\sigma_n$ is at most $\log(t) = 2\lceil \log\log(n) \rceil$.
If $b=0$ then  $\sumv(\textbf{v})=\sigma_n$ (since $\sigma_n$ is the unique location that if set to 1 would
make the vectors corresponding to the locations of 1's   sum to $0^k$).  The conditional entropy in this case is 0.

Summing up all of the conditional entropy contributions gives $3+\lceil \log\log(n) \rceil$. 

\end{proof}

 For our last variant, suppose Alice can communicate to Bob with a $w$-ary alphabet instead of a binary alphabet.  Thus, Alice is streamed a permutation $\sigma$, and when $\sigma_i$ arrives she may write any of the symbols $\{1,\ldots,w\}$ on location $\sigma_i$ in $\textbf{v}$. 
At the last step $b \in \{1,\ldots,w\}$ arrives and Alice must write it in location $\sigma_n$. Bob sees $\textbf{v}$ and has to output a set $J$ 
that contains $\sigma_n$. The cost of the protocol is the maximum size of $J$ over all $\sigma$ and $b$.

We will show that Question \ref{question:main} is false in this setting.
To state our result we need some definitions.
Fix $r>1$ and positive integer $k_0$.  For $n \geq k_0$ define for each integer $j \geq 0$
the function $t_j$ defined on integers $n \geq k_0$.  The function $t_0$
is given by $t_0(n)=n$ for all $n$.  For $j \geq 1$, $t_j$ is defined inductively
$t_j(n)=\max(k_0,\lceil \log_r (t_{j-1}(n)) \rceil)$.  Observe that for $j \geq 2$ we have $t_j(n)=t_{j-1}(t_1(n))=t_1(t_{j-1}(n))$.
Thus $t_j$ depends on parameters $r$ and $k_0$ and is a minor variant of the base $r$ iterated log function.

\begin{theorem}
\label{thm:ary}
For each $j \geq 0$ there is a protocol $\Pi_j$ using the alphabet $\{1,\ldots, 2j+1\}$
that has cost at most $t_j(n)$, where  the parameters needed to define $t_j$ 
are $r=2^{1/4}$ and some sufficiently large $k_0$.
\end{theorem}

For example, for a ternary alphabet the cost of the protocol is $O(\log(n))$ and for a 5-ary
alphabet the cost is $O(\log\log(n))$.
To prove this,
we'll need a few elementary standard facts about error correcting codes.  We include proofs to make the
presentation self-contained.  

\begin{proposition}
\label{symdif}
For each $n \geq 2$ there is  a coloring $\chi_n$ of the subsets of $[n]$ by the set $[n^2]$
such that any two sets that have symmetric difference at most 2 get  different colors.
\end{proposition}

\begin{proof} 
Construct the graph whose vertices are subsets of $[n]$ with two vertices joined by an edge if their symmetric difference
has size 1 or 2.
The degree of any vertex is $n(n+1)/2 < n^2$, and so the graph has a proper coloring with
color set $[n^2]$.
\end{proof}

If $\Sigma$ is a finite alphabet and
$\textbf{s} \in \Sigma^k$, a {\em deletion error} is the removal
of some symbol from the string (shrinking the length by 1).   We need the following (which is much weaker than what is possible, but is all we need.)

\begin{proposition}
\label{deletion code}
There is a $k_0$ such that for all integers $k \geq k_0$
there is a code $C_k$ of size at least $2^{k/2}$ over $\{0,1\}^{k}$ that can correct $\lceil 4 \log_2(k) \rceil $  deletion errors.  
\end{proposition}

We note that the $k_0$ that is needed for Theorem~\ref{thm:ary} will be the $k_0$ provided by this Proposition.

\begin{proof}
We can choose $C_k$ to be a maximal independent set in the graph
on $\{0,1\}^k$ in which two strings $\textbf{x}$ and $\textbf{y}$ are joined
if there is a string $\textbf{z}$ that can be obtained from each of them by at most $\lceil 4\log_2(k) \rceil $  deletions.
If $\Delta$ is the maximum degree of the graph then any maximal independent set has size at least
$2^k/(\Delta+1)$ and $\Delta$ is at most $\binom{k}{\lceil 4 \log_2(k)\rceil}^22^{\lceil 4\log(k)\rceil}$ (since given $\textbf{x}$ 
each neighbor $\textbf{y}$ of $\textbf{x}$
can be constructed by selecting the subset of $\lceil 4\log_2(k)\rceil$ positions to delete from $\textbf{x}$, the subset
of $\lceil 4\log(k) \rceil$ positions to delete from $\textbf{y}$  and the values of the bits deleted from $y$).  For sufficiently large
$k$ this is at most 
$2^{k/2}-1$.  
\end{proof}

\begin{proof}[Proof of Theorem~\ref{thm:ary}]
Fix $k_0$ according to Proposition~\ref{deletion code} and let $r=2^{1/4}$.   Note that $\log_r(n)=4\log_2(n)$. 
Define the functions $t_j$ as above.
 
For $n \leq k_0$ our protocol will just have
Alice write the same symbol every time and Bob output $[n]$.  So assume $n > k_0$.

We prove the theorem by induction on $j$.    
For the induction we need to strengthen the theorem to say that the constructed protocol $\Pi_j$ works in $j+1$ phases
numbered 0 to $j$
where during phase 0, Alice sees $t_0(n)-t_1(n)$ permutation values
and writes only $2j+1$ and during phase $i \in [1,j-1]$ Alice processes the next $t_{i}(n)-t_{i+1}(n)$ permutation values
and writes only symbols $2(j-i)+1$ and $2(j-i)+2$.    During phase $j$, Alice processes $t_j(n)-1$ permutation values
and writes symbols $1$ and $2$.

The protocol $\Pi_0$ is trivial: the alphabet is $\{1\}$ and $t_0(n)=n$ and $t_1(n)=1$.  
Alice writes only 1's. and Bob
outputs the set $[n]$.  

Now suppose $j>1$ and that $\Pi_{j-1}$ has been defined.  Phase 0 of $\Pi_j$ is prescribed.  Let $t=t_1(n)$ and
let $S=\{s_1 < \ldots < s_t\}$ be the unfilled positions after phase 0.    Alice identifies the
set $S$ with the set $[t]$ by the correspondence $s_j \leftrightarrow j$ and views the
remaining $t$ symbols of $\sigma$ as a permutation $\sigma'$ of $[t]$.   The remaining $j-1$ phases of the $\Pi_j$ 
correspond to the protocol $\Pi_{j-1}$ run on $\sigma'$, so Phase $i$ of $\Pi_j$ corresponds to Phase $i-1$
of $\Pi_{j-1}$ run on $\sigma'$.  For $i\geq 2$, Phase $i$ of $\Pi_j$ is exactly the same as Phase $i-1$ of $\Pi_{j-1}$.
However, Phase 1 of $\Pi_j$ is different from Phase 0 of $\Pi_{j-1}$.  In Phase 0 of $\Pi_{j-1}$ the only symbol
written is $2j-1$ but in Phase 1 of $\Pi_j$ both symbols $2j-1$ and $2j$ are used. 
Since $t \geq k_0$, we can construct $C_t$ as in
Proposition~\ref{deletion code} and by changing the alphabet, we can view $C_t$ as a subset
of $\{2j-1,2j\}^t$.  By the choice of $t=\lceil \log_b n \rceil \geq 4\log_2 n$, we have $n^2 \leq 2^{t/2}$ so we can fix  a 1-1 map $g$ from $[n^2]$ to $C_t$.
Alice computes $g(\chi_n(S))$ where
$\chi_n$ comes from proposition~\ref{symdif}.  This is a string $\textbf{y} \in \{2j-1,2j\}^t$ and during phase 1,
Alice write $y_i$ on location $s_i$.  This completes the specification of $\Pi_j$.

We now turn to Bob's strategy for choosing the set $J$ to output.
Let $A_i$ be the set $\{2i+1,2i+2\}$.  During phase $i$, Alice only writes symbols from $A_{j-i}$
so the number of symbols from $A_{j-i}$ written by Alice is $d_i(n)=t_i(n)-t_{i-1}(n)$ if $i<j$ and
is $d_j(n)=t_j(n)-1$ if $i=j$.    The final symbol $b$ comes from some $A_{j-i}$;
let $i^*$ be the index such that $b \in A_{j-i^*}$.

When receiving Alice's output array Bob can count the number of symbols from 
each $A_{j-i}$.  For all but one $i$ this will be $d_i(n)$, and will be $1+d_{i}(n)$ if and only
$i=i^*$.

If $i^* \neq 0$ then Bob knows the set of positions that Alice wrote $2j+1$ to during phase 0,
and therefore knows the set $S$ of $t_1(n)$ positions that remained unfilled at the end
of phase 0.  Since $b<2j+1$, by identifying symbols $2j$ and $2j-1$, Bob
can interpret the array restricted to $S$ as the output of $\Pi_j$ on a set of size $t_1(n)$.
By induction he can determine a set of size at most $t_{j-1}(t(n))=t_j(n)$ that contains $\sigma_n$.

This leaves the case $i^*=1$
Then Bob sees $n-t_1(n)+1$ positions that contain $2j+1$ one of which is $\sigma_n$.
Let $S'$ be the set of positions that don't have $2j+1$ written on them.  Then Bob knows $S'$.
We argue that Bob can recover the set $S$ of positions not written during Phase 0.
From this, Bob will know $\sigma_n$, since $S-S'=\{\sigma_n\}$.

For those positions $s_i \in S$ that Alice wrote during phase 1,
Alice wrote $y_i$ in position $s_i$ where $\textbf{y}=g(\chi_n(S))$.   The number of symbols written
during phase 1 is $t_1(n)-t_2(n)= t-\lceil 4\log_2(t) \rceil$ (unless $j=1$ in which case $t-1$ symbols were written in phase 1).
Thus the string $\textbf{z}$ seen by Bob (using  symbols from $\{2j-1,2j\}$) is obtained from
$\textbf{y}$ with at most $\lceil 4\log_2(t) \rceil$ symbols deleted.
Since $C_t$ is robust against $\lceil 4\log_2(t) \rceil$ deletions, Bob
can recover $\textbf{y}$ from $\textbf{z}$.    He then knows $g^{-1}(\textbf{y})=\chi_n(S)$.  The choice
of $\chi_n$ implies that $S$ is uniquely determined from $S'$ and $\chi_n(S)$,
so Bob recovers $S$ and therefore $\sigma_n$.  
\end{proof}
%{\bf Remark:} The fact that there is such a low cost protocol for ternary alphabets is somewhat troubling if one hopes to prove question \ref{question:main}. However, we note somewhat optimistically that this ternary protocol can be seen to be a generalization of monotone protocols on binary alphabets (for which we proved a sufficient lower bound in section \ref{sec:lowerbounds}). Thus any counterexample to question 1 will likely have to look quite different from the construction given here.

%Consider protocols on alphabet $[k] = \{0,1,\cdots,k\}$. For each location $i$ let $A_i^{(1)}, A_i^{(2)},\cdots, A_i^{(k)}$ be monotone sets of partial assignments to the remaining locations. A protocol on alphabet $[k]$ is said to be {\em monotone} if when location $i$ arrives, Alice does the following: She finds the largest $j$ for which the present partial assignment $\textbf{v}$ is in the set $A_i^{(j)}$ and writes a $j$ on location $i$. If $\textbf{v}$ is not in $A_i^{(j)}$ for any $j$ she writes a 0 on location $i$.

%It is easy to check that the ternary protocol presented above is monotone in this sense, and that for $k = 1$ the definition coincides with the original definition of monotone protocols for binary alphabets.

   \section{Lower Bounds for Restricted Protocols} \label{sec:lowerbounds}
  In the previous section, two stronger variants of Question~\ref{question:main} turned out to have negative answers, which  may suggest that Question~\ref{question:main} also has a negative answer. 
In this section however, we prove a lower bound which implies that any counterexample to Question~\ref{question:main} will need to look quite different from the two protocols provided in the last section. 
  
An order oblivious protocol can be specified by a sequence of maps $A_1,\cdots,A_n$ where each $A_i$ maps partial assignments on the set $[n]$ to a single bit. When location $\sigma_i$ arrives, the bit Alice writes is  $A_{\sigma_i}(\textbf{v})$. 
%Consider the following partial order on partial assignments in $\{0,1,*\}^n$. 
For partial assignments $\alpha$ and $\beta$, we say that $\beta$ is an {\em extension} of $\alpha$, denoted as $\beta \geq \alpha$, if
$\beta$ is obtained from $\alpha$ by fixing additional variables.
An order oblivious protocol is {\em monotone} if each of the maps $A_1,\cdots,A_n$ are monotone with respect to the extension partial order. That is, if $\beta \geq \alpha$ are partial assignments, then $A_{i}(\beta) \geq A_{i}(\alpha)$ for each $i$. As a remark, when running the protocol there may be assignments that are never written on $\textbf{v}$, however defining each $A_i$ to have domain all partial assignments is still valid and simplifies notation.

  Both the AND-OR protocol described in the introduction and the protocol constructed in Theorem \ref{thm:construction} are examples of monotone protocols. Monotonicity generalizes to protocols on $w$-ary alphabets, and the $w$-ary protocol
of Theorem~\ref{thm:ary} is  monotone (if we order the alphabet
symbols in reverse $2j+1<2j<\cdots < 1$). Our main result in this section is that monotone protocols on binary alphabets have cost at least $\lfloor\sqrt{n} \rfloor$. In particular, Question \ref{question:main} is true for such protocols. For the rest of the paper, all protocols will be on binary alphabets.

  \begin{theorem} \label{thm:monotone}
  All monotone protocols have cost at least $\lfloor\sqrt{n} \rfloor$.
  \end{theorem}
  \begin{proof}
  Let $\Pi$ be a monotone protocol.   We show that $E(\Pi)$ has a vertex of degree at least $\lfloor \sqrt{n} \rfloor$.
     
    For a permutation $\sigma$ denote by $\bump{\sigma}{k}$ the permutation obtained from $\sigma$ by ``bumping'' the element $k$ to the end of $\sigma$ and maintaining the same relative order for the rest of $\sigma$. For example, $\bump{321654}{1} = 326541$.
    
    We let $w(\sigma)$ denote the array $\Pi_A(\sigma)$  with the entries sorted in $\sigma$ order. In other words, $w(\sigma)$ is the array defined by $w(\sigma)_i = \Pi_A(\sigma)_{\sigma_i}$.   
    \begin{claim} \label{claim:0s}
       Let $\sigma$ be any permutation and let $\tau$ be obtained from $\sigma$ by performing some sequence of bumps on $\sigma$. Suppose that $\tau$ and $m<n$ satisfy the following:
       \begin{itemize}
       \item The elements $\tau_1,\tau_2,\cdots,\tau_m$ were never bumped.
       \item Alice originally wrote a 0 on the locations $\tau_1,\cdots,\tau_m$, that is $\Pi_A(\sigma)_{\tau_i} = 0$ for all $i \leq m$.
       \end{itemize}
     Then $\Pi_A(\tau)_{\tau_i}$ = 0 for all $i \leq m$, i.e., $w(\tau)$ begins with $m$ 0's.
    \end{claim}
    \begin{proof}
    The claim follows easily by induction on $i$. Suppose we have already shown that $w(\tau)$ begins with $(i-1)$ 0's. Let $\textbf{v}(\sigma,k)$ denote the partial assignment written on $\textbf{v}$ just before Alice receives the index $k$ (here the reader should
take care to distinguish this from the partial assignment just before Alice receives $\sigma_k$).
 Consider the partial assignment $\textbf{v}(\tau,\tau_i)$. It follows from the first assumption and the inductive hypothesis that $\textbf{v}(\sigma,\tau_i)$ is an extension of $\textbf{v}(\tau,\tau_i)$. Thus, since Alice originally wrote a 0 on location $\tau_i$, by monotonicity she continues to write a 0 on that location when being streamed $\tau$ (that is $\Pi_A(\tau)_{\tau_i} = 0$).
    \end{proof}

Let $\sigma$ be the permutation such that
$w(\sigma)$ is lexicographically minimum.

\begin{claim} $w(\sigma)$ consists of a block of  0's, followed by a block of 1's,  followed by a single *.
\end{claim}
\begin{proof}
The result is trivial if their are no 1's.  Let $j$ be the location of the first 1, and let $k$ be the last position
in the block of 1's beginning at $j$.  We claim $k=n-1$.  Suppose $k<n-1$.  Then there is a 0 in
position $k+1$.  Let $\tau$ be obtained from $\sigma$  by bumping $\sigma_j,\ldots,\sigma_k$.
By  Claim ~\ref{claim:0s}, $w(\tau)$ begins with $j$ 0's, contradicting the lexicographic minimality of $\sigma$.
\end{proof}

Let $n-t$ be the number of initial 0's in $w(\sigma)$ so the number of 1's is$t-1$.
Let $T=\{\sigma_{n-t+1},\ldots,\sigma_n\}$ and let $x$ be the vector that is 1 in those positions and 0 elsewhere.
For $k$ between 1 and $n$, let $\tau^{(k)}=\bump{\sigma}{k}$, so $\tau^{(\sigma_n)}=\sigma$.  

The vectors of the form $\Pi_A(\phi)$ and $w(\phi)$ have a single *.  For $b \in \{0,1\}$ we write
$\Pi_A(\phi,b)$ and $w(\phi,b)$ for the vectors obtained by replacing the * by $b$.

\begin{claim}
The vertices $\Pi_A(\tau^{(k)},1)$ for $k \in T$ are all equal to $x$. Therefore $x$ belongs to an edge in direction
$k$ for each $k \in T$ and so has degree at least $t$ in $E(\Pi)$.
 \end{claim}

 \begin{proof}
Let $k \in T$.  Clearly $w(\tau^{(k)},1)$ has the first $n-t$ bits 0, and so by the choice of $\sigma$ the remaining bits are 1.
This implies $\Pi_A(\tau^{(k)}$ has 1's in the positions indexed by the last $t$ elements of $\tau^{(k)}$ which is the set $T$.
\end{proof}

To conclude the proof of the theorem we will find an assignment $y$ that has degree at least $(n-t)/(t+1)$ in the graph $E(\Pi)$.

\begin{claim}
 For $k$ among the first $n-t$ elements of $\sigma$,  $w(\tau^{(k)})$
has the first $n-t-1$ bits equal to 0, and has at most one 0 among the next $t$ bits (and last bit *).
\end{claim}
\begin{proof}
 Claim~\ref{claim:0s} immediately implies that the first $n-t-1$ bits of $w(\tau^{(k)})$ are  0. 
Now take all of the locations that are labeled 1 in $\Pi_A(\tau^{(k)})$ and bump them to the end and let this new permutation be $\rho$. 
Claim~\ref{claim:0s} implies that all 0's remain 0. By the lexicographic minimality of $w(\sigma)$, $w(\rho)$ has  at most $n-t$ 0's
which implies that there was at most a single 0 in $\tau^{(k)}$ in positions $n-t+1$ or higher. 
\end{proof}

Now classify each of the first $n-t$ elements of $\sigma$ into sets $S_{n-t},\ldots,S_n$.
Element $k \in S_n$ if $w(\tau^{(k)})$ has $t$ 1's. 
Otherwise $k \in S_j$ where $j$ is the location of the unique 0 of $w(\tau^{(k)})$  in locations $n-t$ to $n-1$.
Choose $j^*$ so that $|S_{j^*}|$ is maximum and let $m=|S_{j*}|$, which is
at least $(n-t)/(t+1)$.   For $k \in S_{j^*}$, let $y^{ki)}=\Pi_A(\tau^{(k)},0)$.   Let $u=\sigma_{j^*+1}$ and let $y$ be
the vector that is 1 on the positions of $T-\{u\}$ and 0 elsewhere.
\begin{claim}
  The assignments $y^{(k)}$  for $k \in S_{j^*}$ are all equal to $y$, and thus $y$ has degree at least $m$ in $E(\Pi)$.
\end{claim}
\begin{proof}
By the definition of the bump operation the sequence of elements appearing in positions $n-t,\ldots,n-1$ in $\tau^{(k)}$ is $\sigma_{n-t+1},\ldots,
\sigma_{n}$ and the element in position $j^*$ of $\tau^{(k)}$ is $\sigma_{j^*+1}=u$.
Thus $y^{(k)}$ is 1 on the elements of $T-\{u\}$ and 0 elsewhere.
 \end{proof}

We thus have a point $x$ of degree at least $t$ and a point $y$ of degree at least $(n-t)/t+1$ in $E(\Pi)$.
This implies that cost of $\Pi$ is at least $\max(t,(n-t)/(t+1)) > \sqrt{n}-1$ and is thus at least $\lfloor \sqrt{n} \rfloor$.
\end{proof}

   As demonstrated by the AND-OR protocol, Theorem \ref{thm:monotone} is tight up to a constant factor. We remark that the monotone protocols we consider here seem to have no general connection to the class of  monotone boolean functions, and our result for monotone protocols seems to be unrelated to the easy and well known fact that the sensitivity conjecture is true for monotone functions.
   
   We conclude this section with a lower bound for a second class of protocols. Although the lower bound is only logarithmic,  proving a logarithmic lower bound for all protocols with a large enough constant would improve on the best known bounds
relating degree and sensitivity.

We need a few definitions. Recall that an edge $e \in \edges$ may be written as an array in $\{0,1,*\}^n$ for which $e_\ell = *$ on exactly one location $\ell$. We call this location $\ell$  the {\em free location} of that edge. We say two edges $e,e'$ {\em collide} if $e_\ell = e'_\ell$ for all $\ell$ that is not a free location of either edge. Equivalently, two edges collide if they share at least one vertex (each edge collides with itself). Both of the lower bounds in this section will follow by finding an edge $e \in E(\Pi)$ that collides with $m$ other edges in $E(\Pi)$. This implies at least one of the vertices in $e$ has degree at least $m/2$ in the graph $E(\Pi)$, which in turn lower bounds the cost of the protocol.

For a permutation $\sigma$, we write $\ell <_\sigma k$ to denote that the element $\ell$ comes before the element $k$ in $\sigma$. 
Let $\precset{\sigma}{k}=\{\ell: \ell <_\sigma k\}$. For example, if $\sigma = 321654$ then $S_1(\sigma) = \{2,3\}$. We say a protocol is {\em assignment oblivious} if the bit written by Alice in location $k$ only depends on the set $\precset{\sigma}{k}$ (and not on the assignment of bits to that set). Such protocols can be described by a collection of $n$ hypergraphs $H_1,H_2,\cdots, H_n$, where each $H_\ell$ is a hypergraph with vertex set $[n] \setminus \{\ell\}$. 
When $k$ arrives, Alice writes a $1$ if and only if the set $\precset{\sigma}{k}$ is in $H_{k}$.
  
    %Although the lower bound is relatively weak, proving the same lower bound for all protocols would imply that $deg(f) \leq 4^{s(f)}$ for all boolean functions which would be the best known bound relating sensitivity and degree.
   \begin{theorem} Every assignment oblivious protocol $\Pi$ has $\costpi \geq \log_2(n)/2$.
   \end{theorem} \label{thm:oblivious}
   \begin{proof}
   
   Let $\Pi$ be an assignment oblivious protocol.
     
    Given a permutation $\sigma = \sigma_1\sigma_2\cdots\sigma_n$ and $k \in [n]$ we define $\swap{\sigma}{k}$ to be the permutation obtained by swapping the positions of the elements $k$ and $\sigma_n$ within $\sigma$ and keeping every other element in the same place. For example, $\swap{654321}{3} = 654123$.  The lemma will follow by constructing a permutation $\sigma$ such that  that $\pisig$ and  $\Pi_A(\swap{\sigma}{k})$ collide for each $k \in \{\sigma_{n-1},\cdots, \sigma_{n-\ceil{\log_2(n)}}\}$
    
    We build up such a $\sigma$ in a greedy manner. We start with setting $\sigma_{n-1} = 1$.  With $\sigma_{n-1}$ fixed, the bit Alice writes in location $1$ is completely determined by $\sigma_n$ (and does not depend on the values we later choose for $\sigma_1,\cdots,\sigma_{n-2}$). This holds by the assignment oblivious property and because $\precset{\sigma}{1} = \{\ell : \ell \neq 1, \sigma_n\}$. Let $R_1$ be the locations $\ell$ for which setting $\sigma_n = \ell$ results in Alice writing a $1$ in location 1. At least one of $|R_1|, |R_1^c|$ are bigger than $\ceil{(n-1)/2}$, let $T_1$ be that set. Now we fix $\sigma_{n-2}$ to be any element in $T_1$.

Having fixed $\sigma_{n-1}$ and $\sigma_{n-2}$, the bit Alice writes on location $\sigma_{n-2}$ also only depends on the value of $\sigma_n$. Now let $R_2$ be the subset of indices $j$ in $T_1$ such that setting $\sigma_n=j$ would cause  Alice to write a 1 in location $\sigma_{n-2}$. At least one of $|R_2|, |R_2^c|$ are bigger than $\ceil{(|T_1|-1)/2}$, let $T_2 \subseteq T_1$ be that set. This process is iteratively repeated. At step $i$ we set $\sigma_{n-i}$ to be an arbitrary element of $T_{i-1}$. With $\sigma_{n-1},\cdots,\sigma_{n-i}$ now fixed, the value written in location $\sigma_{n-i}$ depends only on the value of $\sigma_n$. The set $R_i$ is defined to be all such values of $\sigma_n$ that result in Alice writing a 1 in location $\sigma_{n-i}$ and $T_i \subseteq T_{i-1}$ is defined to be the larger of $|R_i|$ and $|R_i^c|$. We proceed until the set $T_i$ has only one element in it, in this case we assign $\sigma_n$ to be that element. This process will take at least $\ceil{\log_2(n)}$ steps.  We then assign the remaining elements to $\sigma_1,\cdots,\sigma_{n-i-1}$ in an arbitrary order. 
    
    We now claim that $\pisig$ and $\Pi_A(\swap{\sigma}{k})$ collide for $k = \sigma_{n},\sigma_{n-1},\cdots,\sigma_{n-\ceil{\log_2(n)}}$. 
    
   \begin{claim}
       
    Let $i < \ceil{\log_2(n)}$, and let $k = \sigma_{n-i}$. Then $\pisig_\ell = \Pi_A(\swap{\sigma}{k})_\ell$ for all $\ell \neq k,\sigma_n$.
    \end{claim}
    \begin{proof}
    Let $\sigma' = \swap{\sigma}{k}$. 
 If $\ell <_{\sigma} k$ then  $\precset{\sigma}{\ell} = \precset{\sigma'}{\ell}$ and so Alice writes the same bit
to location $\ell$ under both permutations.
    
Suppose that $\ell >_{\sigma} k$. Let $j$ be such that $\sigma_{n-j} = \ell$. Note that $\sigma_{n-1} = \sigma_{n-1}', \cdots, \sigma_{n-j} = \sigma_{n-j}'$. Recall that holding $\sigma_{n-1},\cdots,\sigma_{n-j}$ fixed, the bit Alice writes at location $\ell$ depends only on the value of $\sigma_n$, and furthermore that bit is the same as for all settings of $\sigma_n \in T_j$. Since both $\sigma_n$ and $\sigma_n' = k$ are in the set $T_j$, it follows that $\pisig_{\ell} = \Pi_A(\sigma')_{\ell}$. 
    \end{proof}
    
    By the above claim, $\sigma$ collides with $\swap{\sigma}{k}$ for at least $\ceil{\log_2(n)}$ values of $k$. Furthermore, at least one of the vertices in $\pisig$ has degree more than $\ceil{\log_2(n)/2}$. This concludes the proof.
    
   \end{proof}
   
   \section{A Protocol with Lower Cost than the AND-OR Protocol} \label{sec:bestknown}

The AND-OR protocol has cost $\lceil \sqrt{n} \rceil$ which  matches our lower bound for
monotone protocols (within 1).  
In this section we show that non-monotone protocols can give at least a small advantage:

\begin{theorem}
\label{upper bound}
For some $\varepsilon>0$ and all sufficiently large $n$ there is a protocol $\Pi$ win
$\costpi \leq (1-\varepsilon)\sqrt{n}$.
\end{theorem}

\begin{proof}
 The construction is a variant of the AND-OR protocol.  

An {\em $(n,m)$ proper code} is an indexed family
$\{\textbf{x}_S \in \{0,1\}^n | S \in \binom{[n]}{m}\}$ of vectors such that the support of $\textbf{x}_S$ is a subset of $S$.
We need the following fact: For $n$ sufficiently large and  $n \geq k^2 \geq .8 n$ there is an $(n,k^2)$-proper code in which
any two codewords are at hamming distance at least $2k+1$.  (The routine proof of this is given below.)
Choose the least $k$ such that $k^2 \geq .8n$ and construct such an $(n,k^2)$-proper code. 

Protocol $\Pi$ is as follows: Alice writes 0 in the  first $n-k^2$ locations.   Let $S$ be the set of remaining $k^2$ locations.
View $S$ as split into $k$ blocks where the each successive block consists of the smallest $k$ unassigned indices in $S$.
For the last $k^2$ elements of the permutation, when index $j$ arrives Alice writes $\textbf{x}_{S,j}$
unless $j$ is the final element of its block to arrive, in which case Alice writes $1-\textbf{x}_{S,j}$.   

The word received by Bob differs from $\textbf{x}_{S}$ in at most $k$ places (one for each block) and
so by the distance property of the code, Bob can deduce the set $S$.  If there is a block
of $S$ such that the received vector agrees with $\textbf{x}_S$ on the entire block then Bob
outputs that block (since that block must include $\sigma_n$); otherwise Bob outputs the set of positions (one per block) in which
the received vector disagrees with $\textbf{x}_S$ (which again must include $\sigma_n$).

 Finally we prove the existence of the desired $(n,k^2)$-proper code using a standard random construction.
for each $S \in \binom{[n]}{k^2}$ define $\textbf{x}_S$ to be a random vector supported on $S$.   Call a pair of sets $S,T \in \binom{[n]}{k^2}$ bad
if $\textbf{x}_S$ and $\textbf{x}_T$ differ in at most  $2k+1$ positions.  The number of coordinates on
which $\textbf{x}_S$ and $\textbf{x}_T$ differ is at least the number of coordinates in $S$ on which they differ.
Holding $\textbf{x}_T$ fixed we see that this probability that $S,T$ is bad is at most the probability
of fewer than $2k+1$ heads in $k^2$ coin tosses, which is $2^{-k^2(1-o(1))}$.  Taking a union bound
over all pairs of $k$-sets we get that the probability that there is a bad pair is at most $\binom{n}{.2n}^22^{-.8n(1-o(1))}=o(1)$,
and so with positive probability there are no bad pairs, and so the desired code exists.
\end{proof}

As mentioned in the introduction, after a prelimiary version of this paper appeared,
Mario Szegedy \cite{szegedy15} gave a protocol of cost
$O(n^{.4732})$.

  % \section{Lower Bound for Two Stage Protocols} \label{sec:twostage}
%\cite{aaronsonqcc}

\section{Acknowledgements}  
We thank Ran Raz for helpful discussions.
The first author was supported by NSF grant CCF 083727. 
The second author was supported in part by (FP7/2007-2013)/ERC Consolidator grant LBCAD no.~616787, 
a grant from Neuron Fund for Support of Science, and the project 14-10003S of GA \v{C}R. 
The third author was supported by NSF grants CCF-083727 and CCF-1218711, and the Simons Foundation under award 332622.

%\bibliographystyle{abbrv}
%\bibliography{0803}

% that's all folks
\end{document}